\documentclass[aps,twocolumn]{revtex4-1}
\usepackage{amsmath,url,multirow,geometry}                
\usepackage{graphicx}
\usepackage{amssymb}
\usepackage{amsthm}
\usepackage{epstopdf}
\newtheorem{theorem}{Result}

\begin{document}

\title{Interplay between positive feedbacks in the generalized CEV process}
\author{ St. Reimann$^1$   \& V. Gontis$^2$ \& M. Alaburda$^2$}
\email{streimann@ethz.ch}
\affiliation{$^1$Department of Management, Technology and Economics, ETH Zurich\\
Kreuzplatz 5, CH-8032 Zurich, Switzerland\\
$^2$ Institute of Theoretical Physics and Astronomy of Vilnius University\\
A. Gostauto 12,LT-01108 Vilnius,Lithuania
}

\begin{abstract}
The dynamics of the {\em generalized} CEV process $dX_t = aX_t^n dt+ bX_t^m dW_t$   $(gCEV)$ is due to an interplay of two feedback mechanisms: State-to-Drift and State-to-Diffusion, whose degrees are $n$ and $m$ respectively. We particularly show that the gCEV, in which both feedback mechanisms are {\sc positive}, i.e. $n,m>1$, admits a stationary probability distribution $P$ provided that $n<2m-1$, which asymptotically decays as a power law $P(x) \sim \frac{1}{x^\mu}$ with tail exponent $\mu = 2m > 2$. Furthermore the power spectral density obeys $S(f) \sim \frac{1}{f^\beta}$, where $\beta = 2 - \:\frac{1+\epsilon}{2(m-1)}$, $\epsilon>0$. Bursting behavior of the gCEV is investigated numerically. Burst intensity $S$ and burst duration $T$ are shown to be related by $S\sim T^2$.
 \end{abstract}
 \maketitle

The dynamics of the state  $X_t$ of a system which is open to a rapidly fluctuating environment can be described by the non-linear stochastic differential equation
\begin{equation}
dX_t =  f(X_t) \: dt \; + \; g(X_t) dW_t ,
\end{equation} 
$W_t$ the standard Wiener process, under the assumption that 1.) noise enters linear, and 2.) the White Noise approximation is valid, see \cite{HorsthemkeLefever1984}. The drift and the diffusion `coefficients' are depending on the recent state $X_t$ and hence represent `State-to-Drift' or `State-to-Diffusion' feedbacks, respectively. The resulting dynamics and consequently properties like the stationary pdf of the gCEV, the spectral density, and also burst statistics are shown to be due to the interplay between these two feedback mechanisms. \\

The following (informal) argument shows that if both feedback mechanisms $f(X)$ and $g(X)$ have a particular functional relation to each other, given by
\begin{equation}\label{cond}
f(X) = \alpha g(X) g'(X) 
\end{equation}
then their interaction generates a power-law like stationary probability distribution - if it exists - in that 
\begin{equation}\label{PL}
P(x) \sim \frac{1}{g(x)^{2(1-\alpha)}} 
\end{equation}
(Here and in the following the notion $F(x) \sim f(x)$ means that the function $F(x) = f(x)$ for large $x$.)
Note that the proportionality factor $\alpha$ enters the coefficient of the power-law tail. The process  considered in \cite{RuseckasKaulakys2010}, $dX_t = b^2 \left( m - \frac{\lambda}{2}\right) X_t^{2m-1} dt + b X_t^m dW_t$, corresponds to $g(X) = X^m$ and $\alpha = 1-\frac{\lambda}{2m}$ so that the stationary pdf decays as a power-law according to $P(x) \sim \frac{1}{x^\lambda}$.

\section{The generalized CEV process  ({\bf gCEV})} 

In the following we consider a particular setting which is that the drift and diffusion coefficients obey
$f(X_t) = a \: X_t^n$ and $g(X_t) =  b\: X_t^m$. In this case one obtains the Ito diffusion process with positive drift and diffusion parameters $a,b$ given by
\begin{equation}\label{gCEV}
dX_t \;=\; a \: X_t^n  dt \; + \; b X_t^m \: dW_t, \qquad n,m > 0.
\end{equation}
This process is a generalization of the {\sc Constant - Elasticity - of - Variance} model, $dX_t = a X_t \: dt+ b \: X_t^m \: dW_t$, which was originally proposed by Cox, Ingersoll and Ross to describe the dynamics of interest rates in an equilibrium economy and which plays an important role in Mathematical Finance, see references below. The gCEV process, eqn \ref{gCEV} describes dynamics by the superposition of {\em two} different feedback scenarios: 
One is the {\bf State-to-Drift feedback} incorporated in the deterministic part of the gCEV 
\begin{equation}\label{det}
\frac{dX_t}{dt} = a X_t^n,
\end{equation}
while the other one concerns {\bf State-to-Diffusion feedback} due to
\begin{equation}\label{stoch}
dX_t = b X_t^m dW_t
\end{equation}
In the following, we will focus on the case that both dynamical components exhibits {\bf positive feedback} simultaneously in that we require
\begin{equation}\label{nm}
1 \; < \; n,m, \; < \infty!
\end{equation}
In this case gCEV dynamics results from the interplay of two {\sc positive} feedback scenarios, each of which generates self-amplifying, i.e. `explosive' behavior in itself. This is easily seen that the drift term with positive feedback gives rise to a Finite-Time-Singularity, i.e. $X_t \to \infty$ as $t \to t_c$, where $t_c = \frac{1}{n-1} X_0^ {1-n}$, while the solution on the finite interval $[0,t_c]$ is $X_t \propto \frac{1}{(t_c - t)^\frac{1}{n-1}}$.
Positive feedback in the state-to-diffusion term also leads to bursting behavior, in that $X_t$ can attain arbitrary large values while it always remains finite. This follows from the fact that the solution of eqn \ref{stoch} is the inverse power of a d-dimensional Bessel process $X_t \propto \frac{1}{\|{\mathbf B}\|_2^\frac{1}{m-1}}$, where ${\mathbf B}$ is a $d$ dimensional Brownian motion, where $d = \frac{2m-1}{m-1} > 2$, for details see \cite{ReimannSornette2010}. Since therefore $d(m)>2$, it follows from the transitivity of ${\mathbf B}_t$, that ${\mathbf B}_t$ escapes to $\infty$ for $t \to \infty$ slower than $t$ a.s., see \cite{BorodinSalminen2002}, while the origin $0$ is polar, i.e. it will not be touched by ${\mathbf B}_t$. On the other hand ${\mathbf B}_t$ has a positive probability to visit any finite $\epsilon$-neighborhood of the origin before escaping. Consequently the dynamics exhibits arbitrary high but finite excursions. 

\begin{figure}[h]
        \includegraphics[width=3.0in]{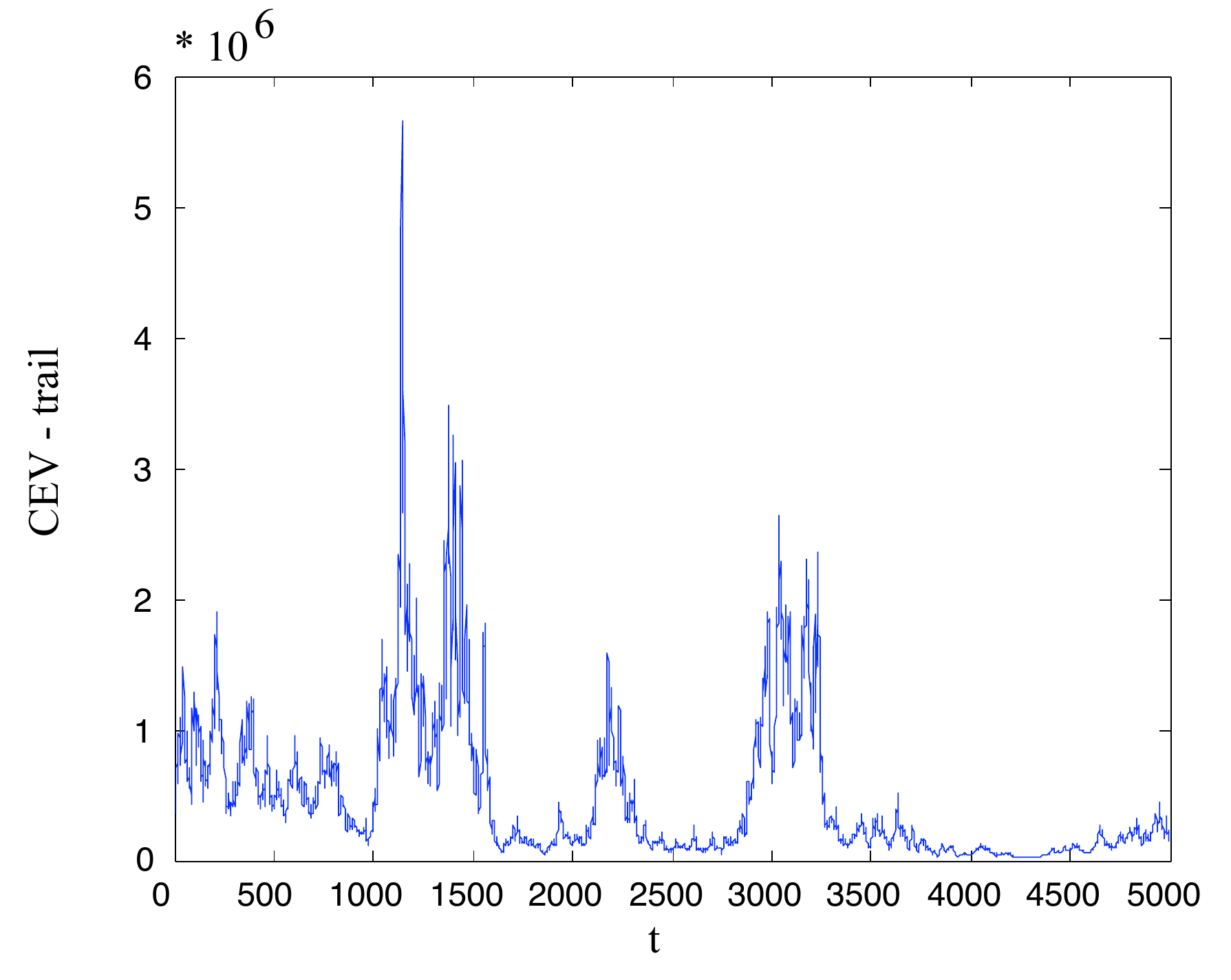}
    \caption{Time trail of the CEV process $dX = b X^\frac{3}{2} dW$}
    \label{fig:CEV_trail} 
 \end{figure}

Hence, for $n,m>1$, both singularities are entirely different: While positive feedback in the state-to-drift component leads to a 'real' Finite-Time-Singularity in that $X_t \to \infty$ within$[0,t_c]$, $X_t$ remains finite even if feedback in state-to-diffusion is positive. Nonetheless, if both positive feedback play in concert, the process exhibits a fat-tailed stationary probability distribution, provided that the State-to-Diffusion feedback is positive $(m>1)$ and strong enough with respect to the State-to-Drift feedback, i.e. $m>\frac{1}{2}(n+1)$. Bursting behavior is reflected in that the stationary pdf decays as a power law $P(x) \sim \frac{1}{x^\mu}$ for large $x$ with an exponent obeying $\mu = 2m > 2$. For large $x$, the tail exponent only depends on the {\em state-to-diffusion} feedback parameter $m$, while the {\em state-to-drift} feedback parameter $n$ determines the pdf only for small $x$.

\section{ Stationary pdf of a generalized CEV process}
\subsection{The CEV process}
\noindent
The standard CEV process is obtained from eqn \ref{gCEV} for $n=1$ 
\begin{equation}\label{CEV}
dX_t \;=\; a \: X_t  dt \; + \; b X_t^m \: dW_t, \qquad a,b > 0,
\end{equation}
A typical time series generated by the CEV process for $m=\frac{3}{2}$ is shown in Fig \ref{fig:CEV_trail}.
An extensive discussion of the CEV process and its relation to other processes including Bessel processes can be found in \cite{JeanblancYorChesney2009, BorodinSalminen2002}. While most of former research on the CEV process has been restricted to the case $0 < m < 1$, we instead focus on the case that state-to-diffusion is subject to positive feedback $m>1$. A detailed discussion of this process and the following theorem as well as the proof can be found in \cite{ReimannSornette2010}. As shown there, a CEV process with $m>1$ is equivalent to a radial Ornstein-Uhlenbeck process for order $\nu(m) = \frac{1}{2(m-1)}$ and hence admits a closed form analytical solution given below:
\begin{theorem}[{\bf Solution of the CEV process for $m > 1$}]\label{MAIN}
The unique and strong solution of the CIR-CEV model, eq \ref{CEV}, with $m>1$ is 
	\begin{equation}
	X(t) \; = \; c(m)  \; { 1 \over {\: \|\mathbf M}(t)\: \|_2^{1/(m-1)}}, 
	\label{X-CIR-CEV}
	\end{equation}
	with $c(m) = \Big(  b (m -1 )\Big)^\frac{1}{1-m}$, where ${\mathbf M(t)}$ be a d-dimensional mean-reverting Ornstein-Uhlenbeck process, whose dimension is a function of the feedback parameter $m$ given by
	\begin{equation}\label{dim}
	\delta(m) \; = \; 2 \; + \; \frac{1}{m-1} \ge 2
	\end{equation} 
	Its components obey $dM_i(t) = -\mu M_i \: dt + dB_i(t)$ with $a\ge 0$ and $B(t)$ the standard Wiener process, while its square norm is $\| {\mathbf M}(t) \|_2$. 
	\end{theorem}
	
The proof is based on the observation that the Lamberti transform of this process \ref{CEV} takes the form of a radial Ornstein Uhlenbeck process of order $\nu(m) = \frac{1}{2(m-1)}$, see \cite{ReimannSornette2010}. In this note we show that the CEV process eqn. \ref{CEV} with positive {\em state-to-diffusion} feedback ($m>1$) admits a stationary probability distribution, which is uni-modal and asymptotically decays as a power-law with its tail exponent proportional to $m$ only.

\begin{theorem}[{\bf Stationary pdf for the CEV process for $m > 1$}]
Let $dX_t = a X_t dt + b X_t^m dW_t$ [CEV] be defined on the non-negative reals $[0,\infty)$ with $a,b>0$ and $W_t$ the Standard Wiener process. Then, if $m > 1$ a stationary probability distribution exists and is similar (not equal) to a Type-2 Gumbel distribution
\begin{equation}\label{ps}
P(x) \; = \;  {\mathcal N} \: x^{-2m} \; e^{- c \: x^{-2(m-1)}},  
\end{equation}
where $c = \frac{2a}{b^2}\frac{1}{2(m-1)} > 0$ and ${\mathcal N} = \frac{2(m-1)}{c^{-\mu} \: \Gamma(\mu)}<\infty$ is a normalization constant with $\gamma = \frac{2m-1}{2(m-1)}>1$. The stationary pdf takes its unique maximum in
$
x_*  =  \left( \frac{b^2}{a} \frac{m }{m-1} \right)^{-\nu(m)}
$
where $\nu(m) = \frac{1}{2(m-1)}$ is the index of radial Ornstein Uhlenbeck process equivalent to eqn. \ref{CEV}. 
\end{theorem}
\begin{figure}[h]
        \includegraphics[width=2.5in]{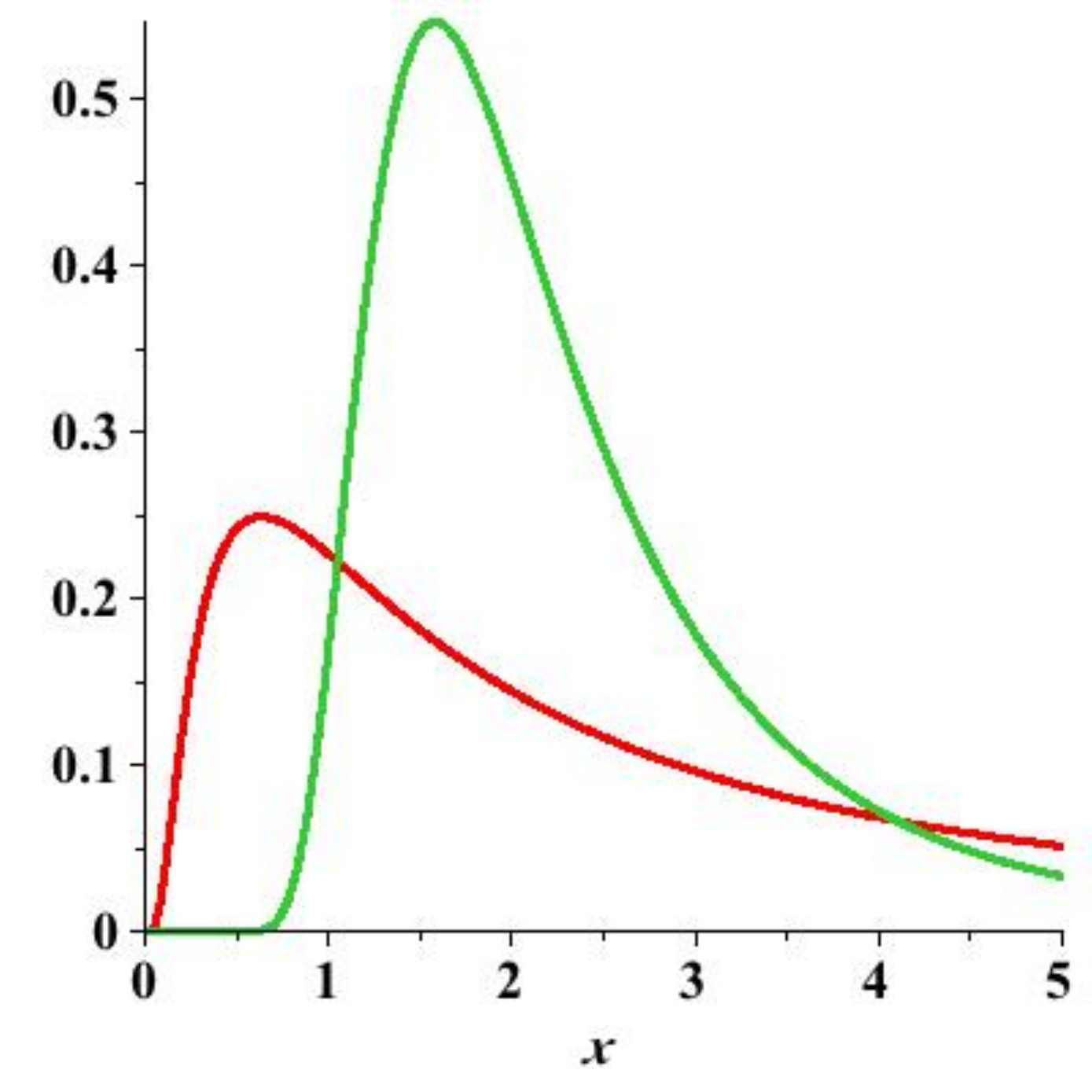}
                   \put(-205,97){\large $\bf P(x)$}
		\put(-150,100){\large $\bf m=\frac{3}{2}$}
		\put(-90,170){\large$\bf m=\frac{5}{2}$}
    \caption{The stationary pdf belonging to the Ito process $dX = a X dt + b X^m dW$ for $m=3/2$ (red) and $m=2$ (green), see eqn. \ref{ps}, defined on the non-negative reals.}
    \label{fig:stat_prob} 
 \end{figure}
\begin{proof}
{As shown in \cite{ReimannSornette2010} the solution of the CEV process [CEV] is an inverse power of a radial Ornstein Uhlenbeck process of dimension $d=\frac{2m-1}{m-1}>2$ for $1<m<\infty$, given by $X_t \propto \frac{1}{\|{\mathbf M}\|_2^\frac{1}{m-1}}$, whose components obey $dM_i = -a M_i dt + dW_i$. Since $0< \|{\mathbf M}\|_2<\infty$, $0$ and $\infty$ are natural boundaries for the CEV process with $m>1$, thus the probability current over these boundaries is zero. 
} 
\end{proof}

\noindent  
For $m = \frac{3}{2}$ the stationary probability distribution asymptotically decays as a power law $P(x) \; \sim \; \frac{1}{x^3}$ for large $x$, more precisely 
\begin{equation}\label{3}
P(x) \; = \; \frac{2a}{b^2} \; \frac{1}{x^3} \; e^{-\frac{2a}{b^2} \frac{1}{x}} \; \sim \; \frac{1}{x^3}
\end{equation}
The pdf is shown in Fig. \ref{fig:CEV_stat_PDF}.

\begin{figure}[h]
\hskip -0.5cm
        \includegraphics[width=3.1in]{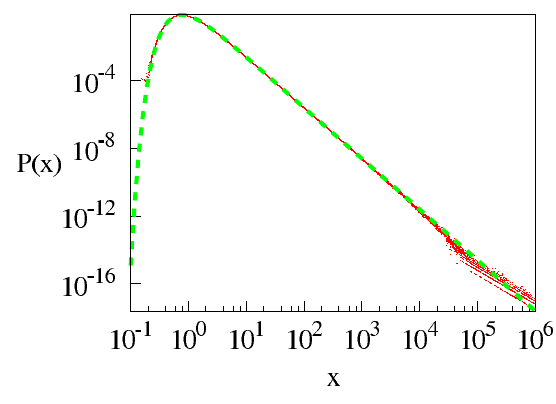}
    \caption{The stationary pdf belonging to the CEV process $dX = a X dt + b X^\frac{3}{2} dW$, with {\sc positive feedback} on `state-to-diffusion. Simulated data in green, while the red dashed curve is due to eqn. \ref{3}.}
    \label{fig:CEV_stat_PDF} 
 \end{figure}

The equivalence between the CEV model with $m>1$ and the rOU process of index $\nu(m) = \frac{1}{2(m-1)}$ shows also up in the functional form of the stationary pdf. In fact the distribution can be rewritten as 
$$
P(x) \; \propto \; x^{-(2+1/\nu(m))} \exp \left[ -\nu(m)  \;  x^{-\frac{1}{\nu(m)}\frac{1}{x}}\right]
$$
which shows that the exponential part of the pdf is controlled by the order of the corresponding rOU process. Furthermore, the stationary probability distribution belonging to the CEV process with $m>1$ is uni-modal for all $b$ and decays for large $x$ as a power-law, i.e. $P(x) \sim\frac{1}{x^{2m}}$. Its graph is sketched in Fig \ref{fig:stat_prob}. As an immediate consequence we have

 \begin{theorem}
Given the CEV process with feedback parameter $m>1$ and let $\nu(m) = \frac{1}{2(m-1)}$. Then
\begin{equation}\label{S}
S(f) \;\sim\;\frac{1}{f^\beta}, 	\quad \beta(m) =  2 - (1+\epsilon) \nu(m),
\end{equation}
where $\epsilon>0$ is a small constant. Consequently $-\infty < \beta(m) < 2$ for $1 < m < \infty$. 
 \end{theorem}
This follows from the more general case, see below. Note that for $m=3/2$ and $c=0$, $\beta =1$, so that in this case $S(f) \sim \frac{1}{f}$, see Fig. \ref{fig:CEV_S(f)}.
\begin{figure}[h]
        \includegraphics[width=3in]{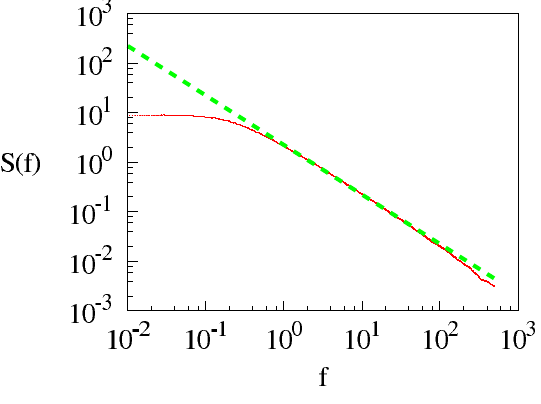}
    \caption{The power spectral density of the CEV process $dX = a X dt + b X^\frac{3}{2} dW$. Simulated density in red, while the green dashed lines is the $1/f$ decay.}
    \label{fig:CEV_S(f)} 
 \end{figure}
\subsection{The generalized CEV process with positive feedback $n,m>1$}

The interplay between state-to-price feedback and state-to-diffusion becomes obvious when considering the Fokker-Planck equation belonging to the gCEV process. Note that by transforming $X_t$ into $Y_t = \frac{1}{b(m-1)} \frac{1}{X^{m-1}}$, the gCEV process becomes $dY_t = \tilde{a}(Y_t) + dW_t$ with drift
{\small
\begin{equation}
\tilde{a}(Y_t) = \frac{1}{b}\; \left[  a \Big( b(m-1) \: Y_t\Big)^\frac{n-m}{1-m}- \frac{m}{m-1} \frac{1}{Y_t}\right]
\end{equation}
}The corresponding potential $U(y) = - \int^y \: \tilde{a}(y') dy'$ of the corresponding Fokker-Planck equation 
is of the form
$$
U(y) \; \propto \; \frac{1}{\eta} \: y^ \eta \; - \; \ln \: y
$$
where $\eta = \frac{n-(2m-1)}{1-m}\ge 0$ if $n\le 2m-1$ and positive otherwise. Thus there is a bifurcation occurring at $n_* = 2m-1$, such that $U(y)$ is convex for $n<n^*$ and concave for $n>n*$, see Fig \ref{fig:U(y)}. Particularly, for $n<n_*$, $0$ and $\infty$ are repelling and the potential has a unique minimum. On the other hand, for $n>n_*$, both $0$ and $\infty$ are attracting, while $U(y)\sim - \frac{1}{{y}^c} \to -\infty$ with some positive $c$ for $y$ approaching $0$. In terms of $X$ this means that for $n>n_*$, that small $X_t \to 0$, while large $X_t \to \infty$. 

\begin{figure}[htbp] 
   \centering
   \includegraphics[width=5.8cm]{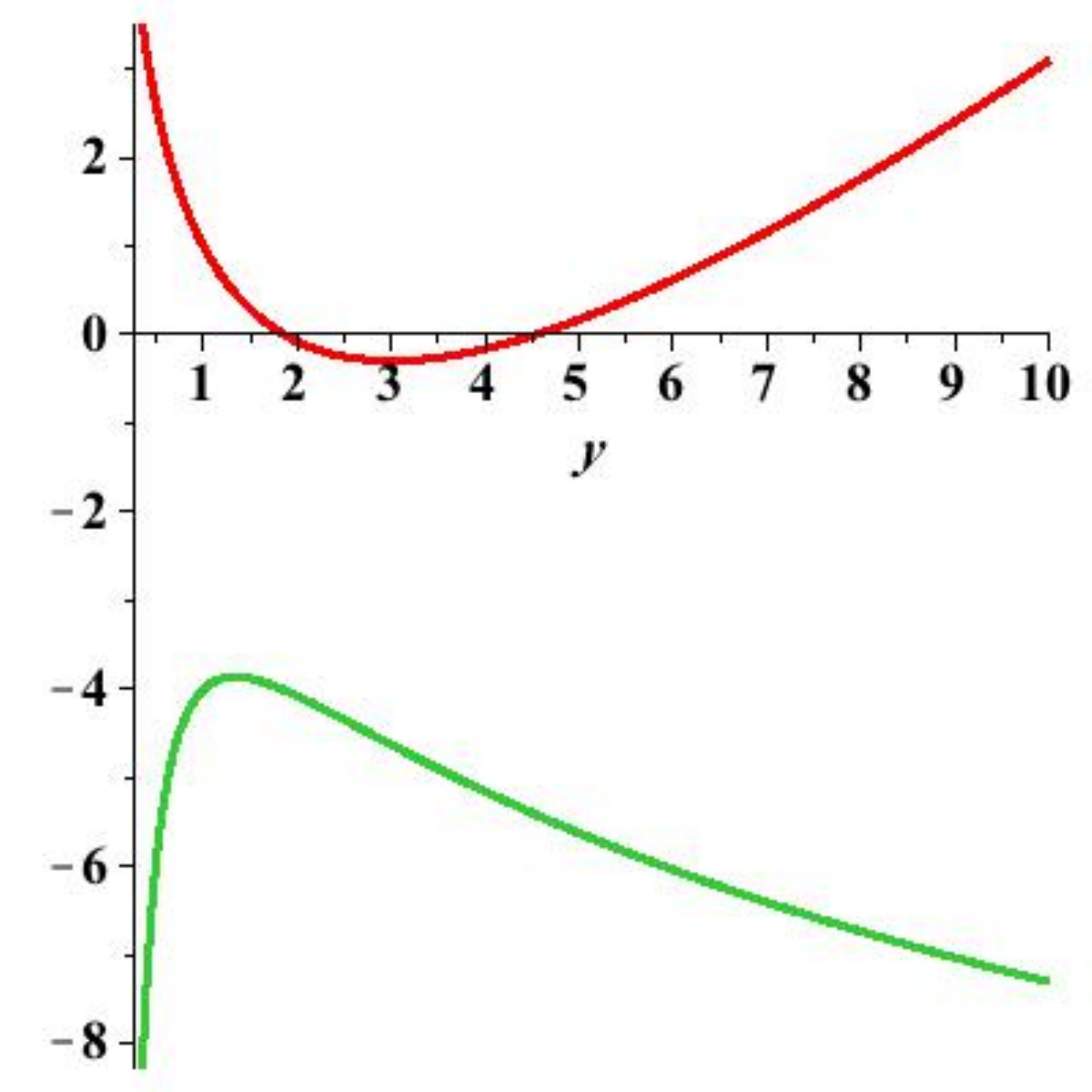}
   \put(-185,97){\large $\bf U(y)$}
   \put(-100,160){\large \bf $n<n_*$}
   \put(-100,60){\large \bf $n>n_*$}
    \caption{Potential $U(y)$ of the Lamperti transformed gCEV process; the upper curve corresponds to $n<n_*$, while the lower one is for $n>n_*$.}
  \label{fig:U(y)}
\end{figure} 

While for $n<n_*$ the gCEV process can be regarded as a diffusion trapped in a convex potential, i.e. behaves locally similar to an Ornstein-Uhlenbeck process, one can expect that in this case it admits a stationary probability density.  

\begin{theorem}
The generalized CEV process admits a stationary probability distribution $P(x)$ if
$m>1$ and $n < 2m-1$. The stationary probability distribution yields
$
P(x) \; = \;  {\mathcal N} \: \frac{1}{x^{2m}} \: \exp\left( -c \: x^{-(2m-n-1)}\right)$, where $c = \frac{2a}{b^2} \frac{1}{2m-n-1} > 0,
$
with normalization constant ${\mathcal N} = \frac{2m+1-n}{c^{-\mu}\Gamma(\mu)}$, $\mu = \frac{2m-1}{(2m-1)-n}$ and thus asymptotically decays as a power law with tail exponent $\mu > 2$
$$
P(x) \; \sim \; \frac{1}{x^\mu},\qquad \mu>2.
$$
\end{theorem}

Note that this excludes the Geometric Brownian Motion case $[n=m=1]$, while it includes the case that the process exhibits positive feedback of both: state-to-drift feedback $[n>1]$ as well as state-to-diffusion $[m>1]$. In fact, given the degree $n>1$ of positive state-to-drift feedback, then the degree of state-to-diffusion feedback $m$ must be positive as well while sufficiently large $m> \frac{1}{2}(n-1)$.\\

Fig \ref{fig: gCEV_stat_PDF} shows the empirical distribution $P$ (red) compared to the analytical solution, see eqn \ref{ps}, (green dashed lines) for the case $m=3/2$ and $n=4/3$, i.e. for the gCEV process in which both partial processes, eqnn \ref{det} and \ref{stoch} exhibit positive feedback. 

\begin{figure}[h]
        \includegraphics[width=2.9in]{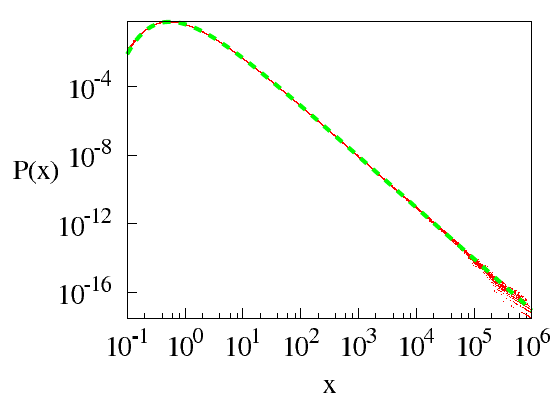}
    \caption{The stationary pdf belonging to the gCEV process $dX = a X^\frac{4}{3} dt + b X^\frac{3}{2} dW$, with {\sc positive feedback} on both `state-to-drift' and `state-to-diffusion.}
    \label{fig: gCEV_stat_PDF} 
 \end{figure}
The power spectrum density of a gCEV asymptotically decays as a power-law with tail index $\beta$, which is a direct function of the index $\nu$ of the related Bessel process and is, particularly independent of the feedback exponent $n$ of the drift term.
\begin{theorem}
The gCVE process eqn \ref{CEV} with feedback parameters $m>1, n<2m-1$ admits a power spectrum  
\begin{equation}\label{S(f)}
S(f) \; \sim \; \frac{1}{f^\beta}, \quad \beta = 2 \; - \; (1+\epsilon) \: \nu(m),
\end{equation}
where $\epsilon>0$ and $\nu(m) = \frac{1}{2(m-1)}$ is the index of the radial Ornstein-Uhlenbeck process equivalent to the CEV process with $m>1$ and $c>0$ is a small parameter.
\end{theorem}
\begin{proof}
{\em The proof follows from a results in \cite{RuseckasKaulakys2010} by noting that the gCEV process for large $X_t$ can be approximated by $dX_t = c(X) X_t^{2m-1} dt + b X_t^m dW_t$, where for $n<2m+1$ the coefficient $c(X) = \frac{a}{m b^2 } X_t^{n-(2m+1)}$ is almost constant for large $X_t$ and approaches $0$ from above if $X_t \to \infty$. Therefore approximating the gCEV process for large but finite $X$ by
\begin{equation}\label{approx}
dX_t \; = \; c X_t^{2m-1} dt \; + \; b X_t^m \: dW_t, \qquad c>0
\end{equation}
we obtain eq. (3) in \cite{RuseckasKaulakys2010} with the substitution $c = \left(m-\frac{\nu}{2}\right)$. According to eqn. (33) in \cite{RuseckasKaulakys2010} the spectral density reads $S(f) \sim \frac{1}{f^\beta}$, $\beta = 1+\frac{\nu-3}{2(m-1)}$. Inserting $\nu = 2(m-c)$ we obtain that for the process in eqn \ref{approx} 
$
\beta(m) = 2 - \frac{1+2c}{2(m-1)} 
$
which gives the result eqn \ref{S(f)} putting $\nu(m) = \frac{1}{2(m-1)}$ and $\epsilon=2c$.
}
\end{proof}
Note that for $m=\frac{3}{2}$ and $c=0$, the power-spectrum shows pure $1/f$ behavior, see Fig. \ref{fig:gCEV_S(f)}, while for $m=\frac{5}{4}$, the spectrum is flat, $\beta=0$. 
\begin{figure}[h]
        \includegraphics[width=3in]{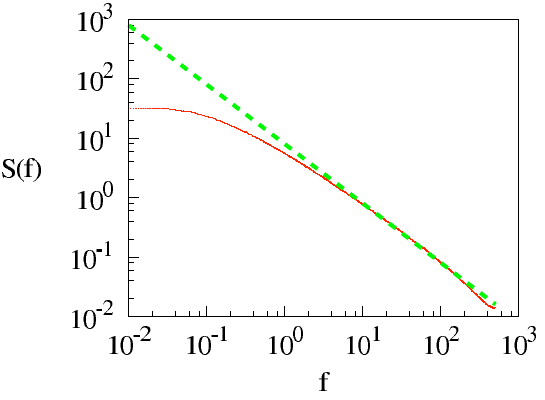}
    \caption{The power spectral density $S(f) \sim \frac{1}{f^1}$ belonging to the gCEV process $dX = a X^\frac{4}{3} dt + b X^\frac{3}{2} dW$.}
    \label{fig:gCEV_S(f)} 
 \end{figure}
The $n$-dependence of the spectral density shows up only for small $f$. One can show that $S(f)$ for small $f$ is an increasing function of $n$. \\

\section{Bursts generated in the gCEV process}

Regarding the transformed gCEV process $Y_t$ for $n<n_*$ as a diffusion being trapped in a convex potential $U(y)$ as in Fig. \ref{fig:U(y)} makes clear that the dynamics of $X_t$ allows for a sequence of arbitrary high but finite outbursts even on short time scales, in agreement with Fig \ref{fig:CEV_trail}. Since the bursting behavior, i.e. $X_t$ large, is governed by the state-to-diffusion feedback parameter $m$, we can restrict ourselves to the case $n=1$ for investigating statistical properties of burst. That is, we will numerically consider the CEV process 
$$
dX_t \; = \; a X_t \: dt + b \: X_t^\frac{3}{2} \: dW 
$$
in the following. Kaulakys and Alaburda \cite{KaulakysAlaburda2009} considered the case $m=2$ in
$dX =\frac{a}{b^2} \; x ^{2m-1} dt + x^{m} dW$, for $x \ge x_m > 0$. Note that is the case $X$ is distributed according to a power law with tail exponent $2\left( m - \frac{a}{b^2}\right)$, as follows from eqn \ref{PL}. In this particular setting they found numerical evidences for clear power-law statistics of bursts. Since in our case, power-law behavior only exists asymptotically, we can expect power-law burst statistics only asymptotically.\\

A {\sc burst} is regarded as a super-threshold event: Let $(X_t)$ be the solution of a gCEV process. The {\em burst interval} $T_k$ of the k-th burst is defined as the time interval between crossing the threshold $\underline{x}>0$ from below and the smallest time at which the threshold is crossed back from above. By a slight abuse of notation we also denote the length of this burst period by $T_k$. In Fig \ref{fig:cev_int} its probability distributions $P_{\underline{x}}(T)=P(T_k = T\:|\: \underline{x})$ are shown for different threshold values: red $ \underline{x} = 2$, green $ \underline{x}=4$, blue $ \underline{x}=8$. The distribution $P(T)$ of burst durations admits an intermediary power-law regime with $\frac{1}{T^\frac{3}{2}}$.  
 \begin{figure}[h]
        \includegraphics[width=3.0in,]{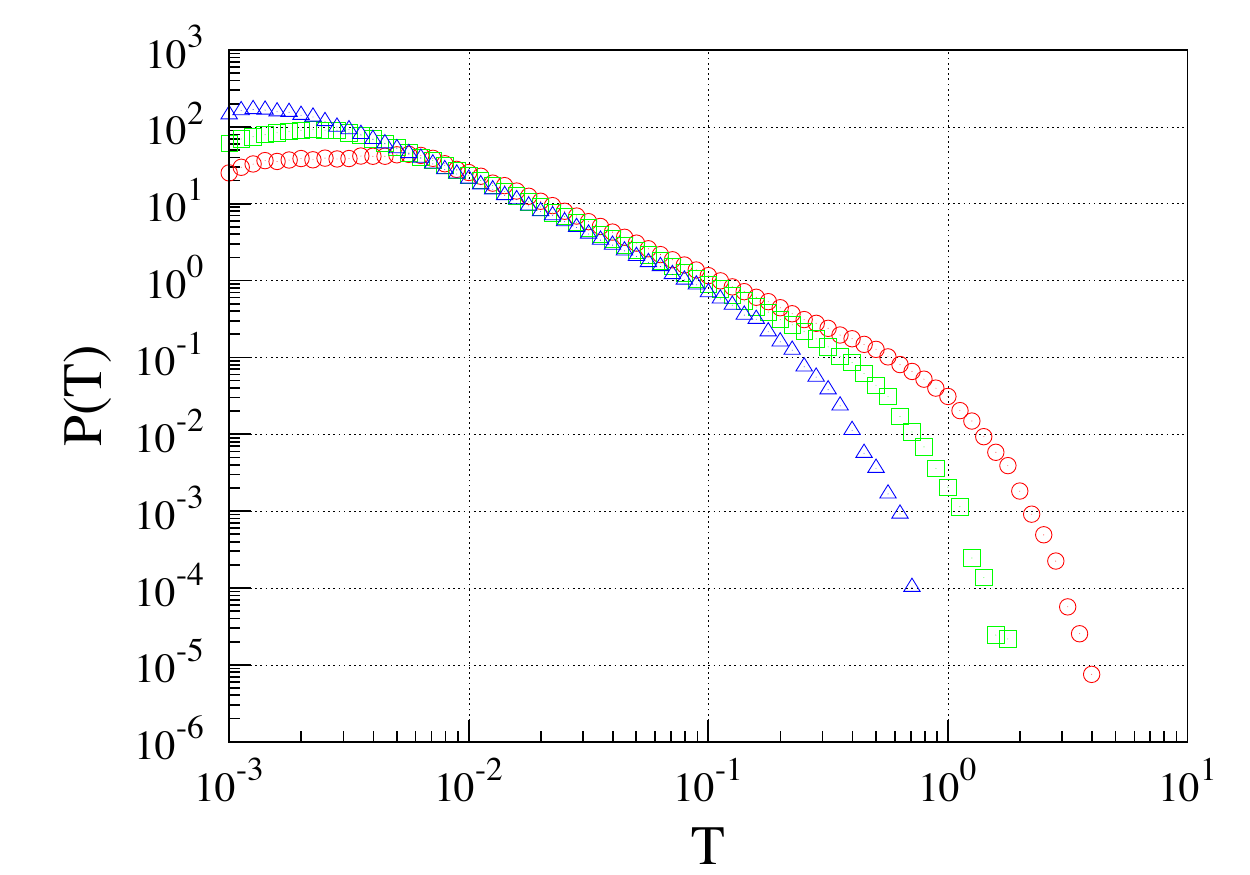}
    \caption{Intermediary power-law regime in $P_{\underline{x}}(T)$ for different threshold values: $\underline{x}=2$ (red circles), $\underline{x}=4$ (green squares), $\underline{x}=8$ (blue triangles).}
    \label{fig:cev_int} 
 \end{figure}

The size of a burst is defined as $S = \int_{T} X_t \: dt$, i.e. the integral over the super-threshold trajectory in the burst period $T$. It turns out to be related to the burst duration $T$ by
 \begin{equation}
 S \; \sim \; T^2 
 \end{equation}
  \begin{figure}[h]
    \includegraphics[width=3.0in]{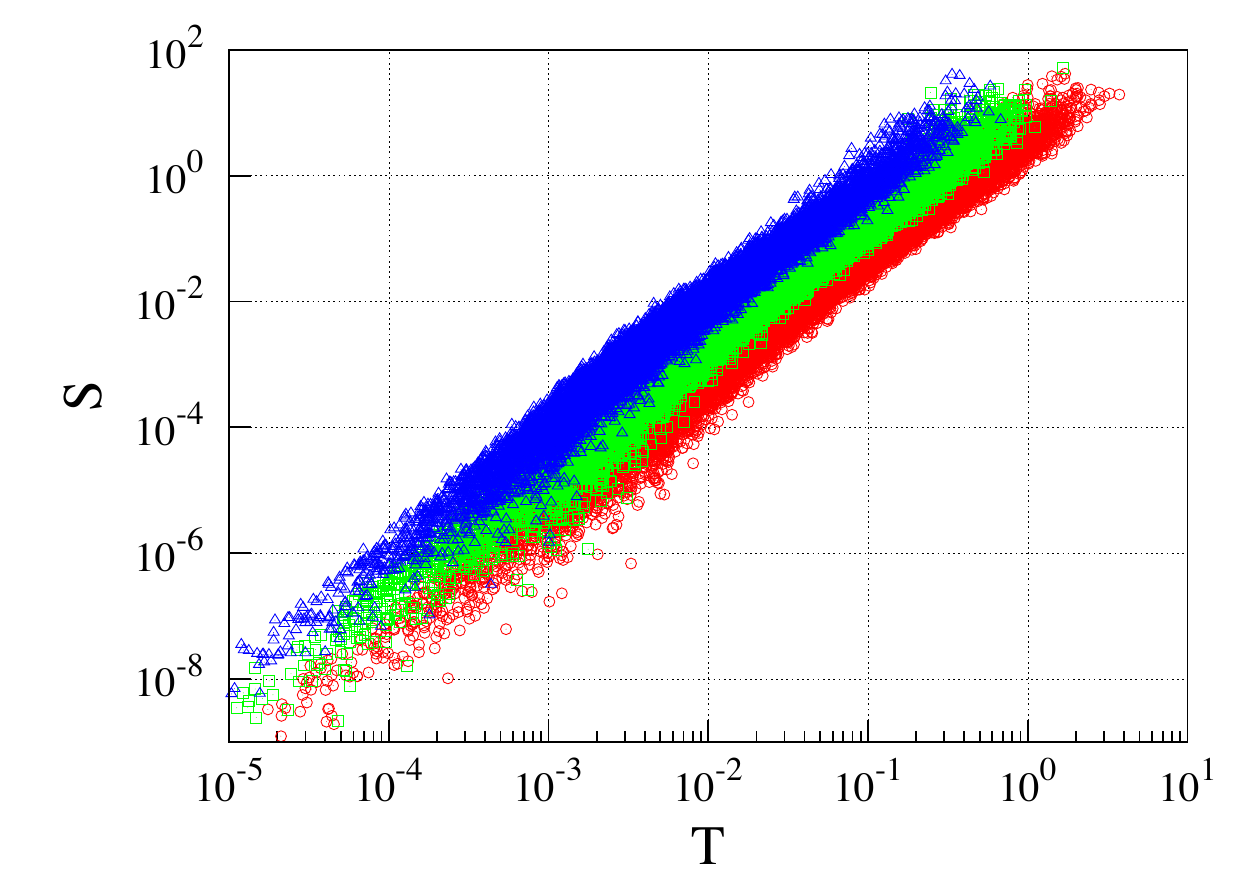}
   \put(-40,150){\large $\underline{x}=4$}
   \put(-130,110){\large \bf $\underline{x}=8$}
   \put(-80,65){\large \bf$\underline{x}=2$}
  \caption{Size $S$ of a burst and its duration $T$ for three thresholds, $\underline{x}=2$ (red circles), $\underline{x}=4$ (green squares), $\underline{x}=8$ (blue triangles).}
  \label{fig:cev_t2s}
 \end{figure}
 
\section{Acknowledgement}
SR is deeply grateful to D. Sornette for numerous helpful discussions about CEV processes and pointing towards the importance of positive feedback. Authors acknowledge the support by EU COST Action MP 0801.


\begin{thebibliography}{1}

\bibitem{BorodinSalminen2002}
A.N. Borodin and P.~Salminen.
\newblock {\em Handbook of Brownian Motion - Facts and Formulae}.
\newblock Birkhaeuer Boston, Basel, Berlin, 2002.

\bibitem{HorsthemkeLefever1984}
Werner Horsthemke and R.~Lefever.
\newblock {\em Noise-induced Transitions: Theory and Applications in Physics,
  Chemistry, and Biology}.
\newblock Springer, 1984.

\bibitem{JeanblancYorChesney2009}
M.~Jeanblanc, M.~Yor, and M.~Chesney.
\newblock {\em Mathematical Methods for Financial Markets}.
\newblock Springer Finance, 2009.

\bibitem{KaulakysAlaburda2009}
B.~Kaulakys and M.~Alaburda.
\newblock Modelling scaled processes and 1/f beta noise using nonlinear
  stochastic differential equations.
\newblock {\em J. Stat. Mech.}, (P02051), 2009.

\bibitem{ReimannSornette2010}
St. Reimann and D.~Sornette.
\newblock Positive feedback in cev processes.
\newblock unpublished, 2010.

\bibitem{RuseckasKaulakys2010}
J~Ruseckas and B.~Kaulakys.
\newblock 1/f noise from nonlinear stochastic differential equations.
\newblock {\em Phys. Rev. E}, 81(031105), 2010.

\end{thebibliography}
\end{document}